\documentclass[aap,MSNbibl,seceqn,dvips]{arximspdf}
\usepackage{breakurl}
\doi{10.1214/12-AAP914} 
\volume{24}
\issue{1}
\pubyear{2014}
\firstpage{54}
\lastpage{75}

\makeatletter

\newtheorem{teor}{Theorem}[section]
\newtheorem{prop}[teor]{Proposition}
\newtheorem{cor}[teor]{Corollary}
\newtheorem{lema}[teor]{Lemma}

\newproclaim{example}[teor]{Example}
\newproclaim{defi}[teor]{Definition}
\newproclaim{rem}[teor]{Remark}

\makeatother

\begin{document}
\begin{frontmatter}

\title{The fundamental theorem of asset pricing, the~hedging problem
and maximal claims in financial markets with short sales prohibitions}
\runtitle{FTAP and hedging without short-selling}

\begin{aug}
\author[A]{\fnms{Sergio} \snm{Pulido}\corref{}\ead[label=e1]{spulido@andrew.cmu.edu}}
\runauthor{S. Pulido}
\affiliation{Carnegie Mellon University}
\address[A]{Department of Mathematical Sciences\\
Carnegie Mellon University\\
5000 Forbes Avenue\\
Pittsburgh, Pennsylvania 15213-3890\\
USA\\
\printead{e1}} 
\end{aug}

\received{\smonth{12} \syear{2011}}
\revised{\smonth{8} \syear{2012}}

%
\begin{abstract}
This paper consists of two parts. In the first part we prove the
fundamental theorem of asset pricing under short sales prohibitions in
continuous-time financial models where asset prices are driven by
nonnegative, locally bounded semimartingales. A key step in this proof
is an extension of a well-known result of Ansel and Stricker. In the
second part we study the hedging problem in these models and connect it
to a properly defined property of ``maximality'' of contingent claims.
\end{abstract}

%

\end{frontmatter}

\section{Introduction}\label{sec1}
The practice of short selling is alleged to magnify the decline of
asset prices. As a result, short sales bans and restrictions have been
commonly used as a regulatory measure to stabilize prices during
downturns in the economy. The most notable recent examples are: (i) in
August of 2011, the European Securities and Markets Authority curtailed
short sales in France, Belgium, Italy and Spain in an effort to stop
the tailspin in the markets caused by the European debt crisis (see
\cite{esma}); (ii) in September of 2008, after the burst of the housing
bubble, the U.S. Securities and Exchange Commission (SEC) prohibited
short selling for 797 financial companies in an effort to stabilize
those companies (see \mbox{\cite{beberpagano,boehmer}}); (iii) at the same
time, in September of 2008, the U.K. Financial Services Authority (FSA)
prohibited short selling for 32 financial companies (see
\cite{beberpagano,boehmer}).

Short sales prohibitions, however, are seen not only after the burst of
a price bubble or during times of financial stress. In certain cases,
the inability to short sell is inherent to the specific market. There
are over 150 stock markets worldwide, many of which are in the third
world. In most of the third world emerging markets the practice of
short selling is not allowed; see \cite{bris,daouk}.
Additionally, in markets such as commodity markets and the housing
market, primary securities such as mortgages cannot be sold short
because they cannot be borrowed.

This paper aims to understand the consequences of short sales
prohibition in semimartingale financial models. The fundamental theorem
of asset pricing establishes the equivalence between the absence of
arbitrage, a key concept in mathematical finance, and the existence of
a probability measure under which the asset prices in the market have a
characteristic behavior. In Section \ref{sec3}, we prove the fundamental theorem
of asset pricing in continuous time financial models with short sales
prohibition where prices are driven by locally bounded semimartingales.
This extends related results by Jouini and Kallal in \cite
{jouinikallal}, Sch\"urger in \cite{schurger}, Frittelli in~\cite
{frittelli}, Pham and Touzi in \cite{phamtouzi}, Napp in \cite{napp}
and more recently by Karatzas and Kardaras in \cite{karatzaskardaras}
to the framework of the seminal work of Delbaen and Schachermayer in
\cite{delbaenmain}.

Additionally, the hedging problem of contingent claims in
markets with convex portfolio constraints where prices are driven by
diffusions and discrete processes has been extensively studied; see
\cite{cvitanickaratzas}, Chapter 5 of \cite{KS} and Chapter 9 of
\cite
{schiedfollmer}. In Section \ref{sec4}, inspired by the works of Jacka in \cite
{jacka} and Ansel and Stricker in \cite{anselstricker}, and using ideas
from \cite{follmerkramkov}, we extend some of these classical results
to more general semimartingale financial models. We also reveal an
interesting financial connection to the concept of maximal claims,
first introduced by Delbaen and Schachermayer in \cite{delbaenmain} and
\cite{delbaennumeraire}.

\section{The set-up}\label{sec2}\label{sectionsetup}

\subsection{The financial market}\label{sec2.1}
We focus our analysis on a finite time trading horizon $[0,T]$ and
assume that there are $N$ risky assets trading in the market. We
suppose, as in the seminal work of Delbaen and Schachermayer in \cite
{delbaenmain}, that the price processes of the $N$ risky assets are
nonnegative locally bounded $P$-semimartingales over a stochastic basis
$(\Omega,\mathcal{F},\mathbb{F},P)$, where $\mathbb{F}:=(\mathcal
{F}_t)_{0\leq t\leq T}$ satisfies the usual hypotheses. We let
$S:=(S^i)_{1\leq i\leq N}$ be the $\mathbb{R}^N$-valued stochastic
process representing the prices of the risky assets. We assume without
loss of generality that the spot interest rates are constant and equal
to 0, that is, the price processes are already discounted. We also
assume that the risky assets have no cash flows associated to them, and
there are no transaction costs.

The probability measure $P$ denotes our reference probability measure.
We suppose that $\mathcal{F}_0$ is $P$-trivial and $\mathcal
{F}_T=\mathcal{F}$. Hence, all random variables measurable with respect
to $\mathcal{F}_0$ are $P$-almost surely constant and there is no
additional source of randomness on the probability space other than the
one specified by the filtration $\mathbb{F}$. As usual, we identify
random variables that are equal $P$-almost surely. If $X$ is a
semimartingale over this stochastic basis, we denote by $L(X)$ the
space of predictable processes integrable with respect to $X$. Given
$H\in L(X)$, $H\cdot X$ denotes the stochastic integral of $H$ with
respect to $X$; see page 165 of \cite{protter}. If $t\in[0,T]$, we let
$\Delta X_t=X_t-X_{t-}$ be the jump of $X$ at time $t$, with the
convention that $X_{0-}=0$. If $\tau$ is a stopping time, $X^{\tau
}:=X_{\cdot\wedge\tau}$ denotes the process $X$ stopped at $\tau$.
Given two semimartingales $X,Y$ we denote by $[X,Y]$ the quadratic
covariation of $X$ and $Y$; see page 66 of \cite{protter}. Given a
probability measure $Q$ equivalent to $P$, denoted by $Q\sim P$, we let
$L^0(Q)$, $L^0_+(Q)$, $L^{\infty}(Q)$, $L^{\infty}_+(Q)$ and $L^1(Q)$
be the spaces of equivalent classes of real-valued random variables,
nonnegative random variables, $Q$-essentially bounded random variables,
nonnegative $Q$-essentially bounded random variables and $Q$-integrable
random variables, respectively. For a measure $Q\sim P$ and a random
variable $f$ bounded from below, we let $E^Q[f]$, $E^Q[f|\mathcal
{F}_t]$ be the expectation with respect to $Q$ and the conditional
expectation with respect to $Q$ given $\mathcal{F}_t$, respectively.
Finally, $\mathcal{H}^1(Q)$ denotes the space of martingales $X$ such
that $E^Q[[X,X]_T^{1/2}]<\infty$.

\subsection{The trading strategies}\label{sec2.2}
We fix $0\leq d\leq N$ and assume that the first $d$ risky assets can
be sold short in an admissible fashion to be specified below and that
the last $N-d$ risky assets cannot be sold short under any
circumstances. This leads us to define the set of admissible strategies
in the market as follows.
%
\begin{defi}\label{admissiblestrategies}
A vector valued process $H=(H^1,\ldots,H^N)$, where for $1\leq i\leq N$
and $t\in[0,T]$, $H^i_t$ denotes the number of shares of asset $i$ held
at time $t$, is called an \textit{admissible trading strategy} if:
\begin{longlist}[(iii)]
\item[(i)] $H\in L(S)$;
\item[(ii)] $H_0=0$;
\item[(iii)] $(H\cdot S)\geq-\alpha$ for some $\alpha> 0$;
\item[(iv)] $H^i\geq0$ for all $i>d$.
\end{longlist}
We let $\mathcal{A}$ be the set of admissible trading strategies.
\end{defi}
Hence, by condition (ii), we assume that the initial risky assets'
holdings are always equal to 0 and therefore initial endowments are
always in num\'eraire denomination. Condition (iii) above is usually
called the \textit{admissibility condition} and restricts the agents'
strategies to those whose value is uniformly bounded from below over
time. The only sources of friction in our market come from conditions
(iii) and (iv) above. For every admissible strategy $H\in\mathcal{A}$
we define the optional process $H^0$ by
%
\begin{equation}
\label{selffinancing}H^0:=(H\cdot S)-\sum
_{i=1}^NH^iS^i.
\end{equation}
If $H^0$ denotes the balance in the money market account, then the
strategy $\overline{H}=(H^0,H)$ is \textit{self-financing} with initial
value 0.

\subsection{No arbitrage conditions}\label{sec2.3}
In \cite{delbaenmain} and \cite{delbaenunbounded}, Delbaen and
Schachermayer considered the no arbitrage paradigm known as no free
lunch with vanishing risk (NFLVR) and proved the fundamental theorem of
asset pricing (FTAP) under this framework. Below we will redefine the
(NFLVR) condition in our context.

Define the following cones in $L^0(P)$:
%
\begin{eqnarray}
\label{setk}
\mathcal{K}:\!&=&\bigl\{(H\cdot S)_T\dvtx H\in\mathcal{A}\bigr
\},
\\
\label{setc}
\mathcal{C}:\!&=&\bigl(\mathcal{K}-L_+^0(P)\bigr)\cap
L^{\infty}(P)
\nonumber\\[-8pt]\\[-8pt]
&=&\bigl\{g\in L^{\infty}(P)\dvtx g=f-h\mbox{ for some $f\in K$ and $h\in
L^0_+(P)$}\bigr\}.\nonumber
\end{eqnarray}
The cone $\mathcal{K}$ corresponds to the cone of random variables that
can be obtained as payoffs of admissible strategies with zero initial
endowment. The cone $\mathcal{C}$ is the cone of random variables that
are $P$-almost surely bounded and are dominated from above by an
element of $\mathcal{K}$. These sets of random variables are cones and
not subspaces of $L^0(P)$ due to conditions (iii) and (iv) in
Definition \ref{admissiblestrategies}. We define in our market the
following ``no arbitrage'' type conditions.
%
\begin{defi}
We say that the financial market satisfies the condition of \textit{no
arbitrage under short sales prohibition} (\textit{NA-S}) if
\[
\mathcal{C}\cap L_+^{\infty}(P)=\{0\}.
\]
\end{defi}
In order to prove the (FTAP), the condition of (NA-S) has to be modified.

\begin{defi}
We say that the financial market satisfies the condition of \textit{no
free lunch with vanishing risk under short sales prohibition} (\textit{NFLVR-S}) if
\[
\overline{\mathcal{C}}\cap L_+^{\infty}(P)=\{0\},
\]
where the closure above is taken with respect to the $\|\cdot\|_{\infty}$
norm on $L^{\infty}(P)$.
\end{defi}
%
\begin{rem}\label{remonnflvr}
Observe that (NFLVR-S) does not hold if and only if there exists a
sequence $(^nH)$ in $\mathcal{A}$, a sequence of bounded random
variables $(f_n)$ and a bounded random variable $f$ measurable with
respect to $\mathcal{F}$ such that $(^{n}H\cdot S)_T\geq f_n$ for all
$n$, $f_n$ converges to $f$ in $L^{\infty}(P)$, $P(f\geq0)=1$ and $P(f>0)>0$.
\end{rem}
In the next section we prove the (FTAP) in our context. This theorem
establishes a relationship between the (NFLVR-S) condition defined
above and the existence of a measure, usually known as the risk neutral
measure, under which the price processes behave in a particular way.

\section{The fundamental theorem of asset pricing}\label{sec3}
The results presented in this section are a combination of the results
obtained by Frittelli in \cite{frittelli} for simple predictable
strategies in markets under convex constraints, and the extension of
the classical theorem of Delbaen and Schachermayer (see \cite
{delbaenmain}) to markets with convex cone constraints established by
Kabanov in \cite{kabanov}. The characterization of \mbox{(NFLVR-S)} is in
accordance with the (FTAP) as proven in \cite{jouinikallal} by Jouini
and Kallal, who assumed that $S_t$ is square integrable under $P$ for
all times $t$ and considered simple predictable strategies.

\subsection{The set of risk neutral measures}\label{sec3.1}
We first define our set of \textit{risk neutral measures}.
%
\begin{defi}\label{esmm}
We let $\mathcal{M}_{\mathrm{sup}}(S)$ be the set of probability measures $Q$ on
$(\Omega,\mathcal{F})$ such that:
\begin{longlist}[(ii)]
\item[(i)] $Q\sim P$ and
\item[(ii)] for $1\leq i\leq d$, $S^i$ is a $Q$-local martingale and,
for $d<i\leq N$, $S^i$ is a $Q$-supermartingale.
\end{longlist}
We will call the set $\mathcal{M}_{\mathrm{sup}}(S)$ the set of \textit{risk
neutral measures} or \textit{equivalent supermartingale measures} (\textit{ESMM}).
\end{defi}
The following proposition plays a crucial role in the analysis below.
%
\begin{prop}\label{mainprop}
Let $\mathcal{C}$ be as in (\ref{setc}). Then
\[
\mathcal{M}_{\mathrm{sup}}(S)=\Bigl\{Q\sim P\dvtx  \sup_{f\in\mathcal{C}}E^Q[f]=0
\Bigr\}.
\]
\end{prop}
To prove this proposition we need the following results.
%
\begin{lema}\label{lemma1}
Suppose that $Q$ is a probability measure on $(\Omega,\mathcal{F})$.
Let $V$ be an $\mathbb{R}^N$-valued $Q$-semimartingale such that $V^i$
is $Q$-local supermartingale for $i>d$, and $V^i$ is a $Q$-local
martingale for $i\leq d$. Let $H$ be an $\mathbb{R}^N$-valued bounded
predictable process, such that $H^i\geq0$ for $i>d$. Then $(H\cdot V)$
is a $Q$-local supermartingale.
\end{lema}
\begin{pf}
Without loss of generality we can assume that, under $Q$, $V^i$ is a
supermartingale for $i>d$. Suppose that for $i>d$, $V^i=M^i-A^i$ is the
Doob--Meyer decomposition of the $Q$-supermartingale $V^i$, with $M^i$
a $Q$-local martingale and $A^i$ a predictable nondecreasing process
such that $A^i_0=0$. Let $M^i=V^i$ and $A^i=0$ for $i\leq d$. Then
$V=M-A$, with $M=(M^1,\ldots, M^N)$ and $A=(A^1,\ldots,A^N)$, is the
canonical decomposition of the special vector valued semimartingale $V$
under $Q$. Since $H$ is bounded, $(H\cdot V)$ is a $Q$-special
semimartingale, $H\in L(M)\cap L(A)$, $(H\cdot V)=(H\cdot M)-(H\cdot
A)$ and $(H\cdot M)$ is a $Q$-local\vadjust{\goodbreak} martingale; see Proposition 2 in
\cite{jacodintegration}. Additionally, since $H^i\geq0$ for $i>d$ we
have that $(H\cdot A)$ is a nondecreasing process starting at 0. We
conclude then that $(H\cdot V)$ is a $Q$-local supermartingale.
\end{pf}
The following lemma is a known result of stochastic analysis which we
present here for completion.
%
\begin{lema}\label{thespaceH1}
Suppose that $H$ is a bounded predictable process, and $X\in\mathcal
{H}^1(Q)$ is a real-valued martingale. Then $H\cdot X$ is also in
$\mathcal{H}^1(Q)$. In particular, $H\cdot X$ is a $Q$-martingale.
\end{lema}
\begin{pf}
The argument to prove this result is analogous to the one used in the
proof of Emery's inequality (see Theorem V-3 in \cite{protter}) and we
do not include its proof in this paper.
\end{pf}
The next proposition is a key step in the extension of the (FTAP) to
markets with short sales prohibition and prices driven by arbitrary
locally bounded semimartingales. It extends a well-known result by
Ansel and Stricker; see Proposition~3.3 in~\cite{anselstricker}.
%
\begin{prop}\label{extofanselandstricker}
Let $Q\in\mathcal{M}_{\mathrm{sup}}(S)$ and $H\in L(S)$ be such that $H^i\geq0$
for $i>d$. Then $H\cdot S$ is a $Q$-local supermartingale if and only
if there exists a sequence of stopping times $(T_n)_{n\geq1}$ that
increases $Q$-almost surely to $\infty$ and a sequence of nonpositive
random variables $\Theta_n$ in $L^1(Q)$ such that $\Delta(H\cdot
S)^{T_n}\geq\Theta_n$ for all $n$.
\end{prop}
\begin{pf}
($\Leftarrow$) It is enough to show that for all $n$, $(H\cdot
S)^{T_n}$ is a $Q$-local supermartingale. Hence, without loss of
generality we can assume that
\[
\Delta(H\cdot S)\geq\Theta
\]
with $\Theta\in L^1(Q)$ a nonpositive random variable. By Proposition 3
in \cite{jacodintegration}, if we define
\[
U_t=\sum_{s\leq t}1_{\{|\Delta S_s|>1\ \mathrm{or}\ |\Delta(H\cdot
S)_s|>1\}
}\Delta
S_s,
\]
there exist a $Q$-local martingale $N$ and a predictable process of
finite variation $B$ such that $H\in L(N)\cap L(B+U)$, $Y:=S-U$ is a
$Q$-special semimartingale with bounded jumps and canonical
decomposition $Y=N+B$ and $H\cdot N$ is a $Q$-local martingale. Let
$V:=B+U$ and $H^{\alpha}:=H1_{\{|H|\leq\alpha\}}$ for $\alpha\geq0$.
We have that $Q\in\mathcal{M}_{\mathrm{sup}}(S)$, $N$ is a $Q$-local martingale
and $V=S-N$. This implies that $V^i$ is a $Q$-local supermartingale for
$i>d$, and $V^i$ is a $Q$-local martingale for $i\leq d$. We can
further assume by localization that $N^i\in\mathcal{H}^1(Q)$ for all
$i\leq N$ and that $V$ has canonical decomposition $V=M-A$, where $M^i$
in $\mathcal{H}^1(Q)$ and $A^i\geq0$ is $Q$-integrable, predictable
and nondecreasing for all $i\leq N$; see Theorem \mbox{IV-51} in~\cite
{protter}. By Lemmas \ref{lemma1} and \ref{thespaceH1}, these
assumptions imply that for all $\alpha\geq0$, $H^{\alpha}\cdot N$ and
$H^{\alpha}\cdot M$ are $Q$-martingales and $H^{\alpha}\cdot V$ is a
$Q$-supermartingale. In particular for all stopping times $\tau\leq T$,
$E^Q[(H^{\alpha}\cdot N)_{\tau}]=0$ and $E^Q[(H^{\alpha}\cdot
V)_{\tau
}]\leq0$. This implies that for all stopping times $\tau\leq T$,
$E^Q[|(H\cdot N)_{\tau}|]=2E^Q[(H\cdot N)_{\tau}^-]$ and
$E^Q[|(H\cdot
V)_{\tau}|]\leq2E^Q[(H\cdot V)_{\tau}^-]$. After these observations,
by following the same argument as the one given in the proof of
Proposition 3.3 in \cite{anselstricker}, we find a sequence of stopping
times $(\tau_p)_{p\geq0}$ increasing to $\infty$ such that
$E^Q[|H\cdot V|_{\tau_p}]\leq12 p+ 4E^Q[|\Theta|]$ and, for all
$\alpha
\geq0$,
$|(H^{\alpha}\cdot V)^{\tau_p}|\leq4p+|H\cdot V|_{\tau_p}.$ An
application of the dominated convergence theorem yields that $(H\cdot
V)^{\tau_p}$ is a $Q$-supermartingale for all $p\geq0$. Since $H\cdot
S=H\cdot N+H\cdot V$ and $(H\cdot N)$ is a $Q$-local martingale, we
conclude that $(H\cdot S)$ is a $Q$-local supermartingale.

($\Rightarrow$) The $Q$-local supermartingale $H\cdot S$ is special. By
Proposition 2 in \cite{jacodintegration}, if $S=M-A$ is the canonical
decomposition of $S$ with respect to $Q$, where $M^i$ is a $Q$-local
martingale, $A_0=0$ and $A^i$ is an nondecreasing, predictable and
$Q$-locally integrable process for all $i\leq N$, then $H\cdot S=H\cdot
M-H\cdot A$ is the canonical decomposition of $H\cdot S$, where $H\cdot
M$ is a $Q$-local martingale and $H\cdot A$ is nondecreasing,
predictable and $Q$-locally integrable. By Proposition 3.3 in \cite
{anselstricker} we can find a sequence of stopping times
$(T_{n})_{n\geq0}$ that increases\vspace*{1pt} to $\infty$ and a sequence of
nonpositive random variables $(\tilde{\Theta}_{n})$ in $L^1(Q)$ such
that
\[
\Delta(H\cdot M)^{T_n}\geq\tilde{\Theta}_{n}.
\]
We can further assume without loss of generality that $(H\cdot
A)_{T_n}\in L^1(Q)$ for all $n$. By taking $\Theta_n=\tilde{\Theta
}_{n}-(H\cdot A)_{T_n}$, we conclude that for all $n$
\[
\Delta(H\cdot S)^{T_n}= \Delta(H\cdot M)^{T_n}-\Delta(H\cdot
A)^{T_n}\geq\tilde{\Theta}_{n}-(H\cdot A)^{T_n}\geq
\Theta_{n}.
\]
\upqed\end{pf}
%
\begin{lema}\label{lemma2}
Let $Q\in\mathcal{M}_{\mathrm{sup}}(S)$ and $H\in\mathcal{A}$; see Definitions
\ref{admissiblestrategies} and \ref{esmm}. Then $(H\cdot S)$ is a
$Q$-supermartingale. In particular $(H\cdot S)_T\in L^1(Q)$ and
$E^Q[(H\cdot S)_T]\leq0$.
\end{lema}
\begin{pf}
Assume that $(H\cdot S)\geq-\alpha$, with $\alpha\geq0$. Let $q\geq
0$ be arbitrary. If we define $T_q=\inf\{t\geq0\dvtx  (H\cdot S)_t\geq
q-\alpha\}$, we have that $\Delta(H\cdot S)^{T_q}\geq-q$. By
Proposition \ref{extofanselandstricker} we conclude that $(H\cdot S)$
is a $Q$-local supermartingale bounded from below. By Fatou's lemma we
obtain that $(H\cdot S)$ is a $Q$-supermartingale as we wanted to prove.
\end{pf}
%
\begin{rem}
The statement of Lemma \ref{lemma2} corresponds to Lemma 2.2 and
Proposition 3.1 in \cite{kallsen}. Here we have proved this result by
methods similar to the ones appearing in the original proof of Ansel
and Stricker in \cite{anselstricker}. Additionally, we have given
sufficient and necessary conditions for the $\sigma$-supermartingale
property (see Definition 2.1 in \cite{kallsen}) to hold.
\end{rem}
We are now ready to prove the main proposition of this section. The
arguments below essentially correspond to those presented in
\cite{delbaenmain,kabanov} and \cite{karatzaskardaras}. We include them
here for completeness.
\begin{pf*}{Proof of Proposition \ref{mainprop}} By Lemma \ref{lemma2}
\[
\mathcal{M}_{\mathrm{sup}}(S)\subset\Bigl\{Q\sim P\dvtx \sup_{f\in\mathcal{C}}E^Q[f]=0
\Bigr\}.
\]
Now suppose that $Q$ is a probability measure equivalent to $P$ such
that $E^Q[f]\leq0$ for all $f\in\mathcal{C}$. Fix $1\leq i\leq N$.
Since $S^i$ is locally bounded, there exists a sequence of stopping
times $(\sigma_n)$ increasing to $\infty$ such that $S^i_{\cdot
\wedge
\sigma_n}$ is bounded. Let $0\leq s<t\leq T$, $A\in\mathcal{F}_s$ and
$n\geq0$ be arbitrary. Consider the process $H^i(r,\omega
)=1_{A}(\omega
)1_{(s\wedge\sigma_n,t\wedge\sigma_n]}(r)$. Let $H^j\equiv0$ for
$j\neq i$. We have that $H=(H_1,\ldots,H_N)\in\mathcal{A}$, $(H\cdot
S)_T\in\mathcal{C}$ and
\[
0\geq E^Q\bigl[(H\cdot S)_T\bigr]=E^Q
\bigl[1_{A}\bigl(S^i_{t\wedge\sigma_n}-S^i_{s\wedge
\sigma_n}
\bigr)\bigr].
\]
This implies that $S^i_{\cdot\wedge\sigma_n}$ is a $Q$-supermartingale
for all $n$ and $S^i$ is a $Q$-local supermartingale. Since $S^i$ is
nonnegative, by Fatou's lemma we conclude that $S^i$ is a
$Q$-supermartingale. For $1\leq i\leq d$ we can apply the same argument
to the process $H^i(r,\omega)=-1_{A}(\omega)1_{(s\wedge\sigma
_n,t\wedge
\sigma_n]}(r)$ to conclude that $S^i$ is a $Q$-local martingale. Hence
\[
\mathcal{M}_{\mathrm{sup}}(S)\supset\Bigl\{Q\sim P\dvtx \sup_{f\in\mathcal{C}}E^Q[f]=0
\Bigr\},
\]
and the proposition follows.
\end{pf*}
We have seen in the proof of this proposition that the following
equality holds.
%
\begin{cor}\label{thesetps}
Let $\mathcal{M}_{\mathrm{sup}}(S)$ be as in Definition \ref{esmm}. Then
%
\begin{equation}
\label{supermeasures}\qquad\mathcal{M}_{\mathrm{sup}}(S)=\bigl\{Q\sim P\mbox
{: $(H\cdot
S)$ is a $Q$-supermartingale for all $H\in\mathcal{A}$}\bigr\}.
\end{equation}
\end{cor}

\subsection{The main theorem}\label{sec3.2}
%
\begin{teor}\label{maintheorem}
(NFLVR-S) $\Leftrightarrow\mathcal{M}_{\mathrm{sup}}(S)\neq\varnothing$.
\end{teor}
In order to prove this theorem we need the following lemma.
%
\begin{lema}\label{lemmapointwiseconstraints}
$\{(H\cdot S)\dvtx  H\in\mathcal{A}, (H\cdot S)\geq-1\}$ is a closed subset
of the space of vector valued $P$-semimartingales on $[0,T]$ with the
semimartingale topology given by the quasinorm
%
\begin{equation}
\label{quasinorm}D(X)=\sup\bigl\{E^P\bigl[1\wedge\bigl|(H\cdot
X)_T\bigr|\bigr]\mbox{: $H$ predictable and $|H|\leq1$}\bigr\}.
\end{equation}
\end{lema}
\begin{pf}
Since $ \{\overrightarrow{x}\in\mathbb{R}^N\dvtx  x^i\geq0$ for
$i>d\}$ is a closed convex polyhedral cone in $\mathbb{R}^N$,
this result follows from the considerations made in \cite{schweizer}.
\end{pf}
%
\begin{rem}\label{polyhedralconstraints}
Notice that for the conclusion of Lemma \ref{lemmapointwiseconstraints}
to hold, it is important to work with short sales constraints as
explained in Definition \ref{admissiblestrategies}. In order to
consider general convex cone constraints an alternative approach is to
consider constrained portfolios modulus those strategies with zero
value. This is the approach taken in \cite{karatzaskardaras}. In our
particular case, and as it is pointed out in \cite{schweizer}, we have
the advantage of considering portfolio constraints defined pointwise
for $(\omega,t)$ in $\Omega\times[0,T]$. Given a particular strategy,
it is easier to verify admissibility when pointwise restrictions are considered.
\end{rem}
\begin{pf*}{Proof of Theorem \ref{maintheorem}}
If $K_1$ and $K_2$ are nonnegative bounded predictable processes,
$K_1K_2=0$, $H_1,H_2\in\mathcal{A}$ are such that $(H_1\cdot
S),(H_2\cdot S)\geq-1$, and $X:=K_1\cdot(H_1\cdot S)+K_2\cdot
(H_2\cdot S)\geq-1$, then associativity of the stochastic integral
implies that $X\in\{(H\cdot S)\dvtx  H\in\mathcal{A}, (H\cdot S)\geq-1\}$.
This fact, Proposition~\ref{mainprop}, Lemma \ref
{lemmapointwiseconstraints} and Theorem 1.2 in \cite{kabanov} imply
that (NFLVR-S) is equivalent to existence of a measure $Q\in\mathcal
{M}_{\mathrm{sup}}(S)$.
\end{pf*}
%
\begin{rem}\label{importantremark}
By using the results obtained by Kabanov in \cite{kabanov}, Karat\-zas
and Kardaras in \cite{karatzaskardaras} proved that the condition of
$(\mathit{NFLVR})$, with predictable convex portfolio restrictions, is
equivalent to the existence of a measure under which the value
processes of admissible strategies are supermartingales. In their work
the set of measures on the right-hand side of equation (\ref
{supermeasures}) is also referred to as the set of equivalent
supermartingale measures. As mentioned in Remark \ref
{polyhedralconstraints}, they considered convex portfolio constraints
modulus strategies with zero value. We have shown that in the special
case of short sales prohibition one can consider pointwise portfolio
restrictions. More importantly, we have shown that in the case of short
sales prohibition, the set of measures under which the values of
admissible portfolios are supermartingales is precisely the set of
measures under which the prices of the assets that cannot be sold short
are supermartingales, and the prices of assets that can be admissibly
sold short are local martingales; see Corollary \ref{thesetps}. This
provides a more precise characterization of the set of risk neutral
measures under short sales prohibition. Given a particular model, this
characterization simplifies the process of verifying that the model is
consistent with the condition of (NFLVR-S).
\end{rem}
This section demonstrates that the results obtained by Jouini and
Kallal in \cite{jouinikallal}, Sch\"urger in \cite{schurger}, Frittelli
in \cite{frittelli}, Pham and Touzi in \cite{phamtouzi} and Napp in
\cite{napp}, can be extended to a more general class of models, similar
to the ones used by Delbaen and Schachermayer in \cite{delbaenmain}. It
is also clear from this characterization that the prices of the risky
assets that cannot be sold short could be above their risk-neutral
expectations at maturity time, because the condition of (NFLVR-S) only
guarantees the existence of an equivalent supermartingale measure for
those prices.

\section{The hedging problem and maximal claims}\label{sec4}
In this section we seek to understand the scope of the effects of short
sales prohibition on the hedging problem of arbitrary contingent
claims. We study, in financial markets with short sales prohibitions
where prices are driven by nonnegative locally bounded semimartingales,
the space of contingent claims that can be super-replicated and
perfectly replicated. The duality type results presented in this
section are robust because they characterize the claims that can be
perfectly replicated or super-replicated in markets with prohibition on
short-selling without relying on particular assumptions on the dynamics
of the asset prices, other than the locally bounded semimartingale
property. By using the results of F\"ollmer and Kramkov in \cite
{follmerkramkov} we extend the classic results of Ansel and Stricker in
\cite{anselstricker}. The results presented also extend those in
Chapter 5 of \cite{KS} and Chapter 9 of \cite{schiedfollmer} to more
general semimartingale financial markets. Additionally, we establish,
in our context, a connection to the concept of maximal claims as it was
first introduced by Delbaen and Schachermayer in \cite{delbaenmain} and
\cite{delbaennumeraire}. The (FTAP) (Theorem \ref{maintheorem}) can be
generalized to the case of special convex cone portfolio constraints
(see Theorem 4.4 in \cite{karatzaskardaras}), and some of the results
presented in this section could be extended to this framework. In our
study, we specialize to short sales prohibition because in this case
the examples are simplified by the fact that the set of risk neutral
measures is characterized by the behavior of the underlying price
processes, rather than the behavior of the value processes of the
trading strategies; see Remark \ref{importantremark}. Additionally, in
this case, the portfolio restrictions can be considered pointwise in
$\Omega\times[0,T]$; see Remarks \ref{polyhedralconstraints} and
\ref
{importantremark}. A~related study on the implications of short sales
prohibitions on hedging strategies involving futures contracts can be
found in \cite{pulidojarrowprotter}. We will use the same notation as
described in Section \ref{sectionsetup}. We will denote by $\mathcal
{M}_{\mathrm{loc}}(S)$ the set of measures equivalent to $P$ under which the
components of $S$ are local martingales.

\subsection{The hedging problem}\label{sec4.1}
This section shows how the results obtained by F\"ollmer and Kramkov in
\cite{follmerkramkov} extend the usual characterization of attainable
claims and claims that can be super-replicated to markets with short
sales prohibition. These results extend those presented in Chapter 5 of
\cite{KS} and Chapter 9 of \cite{schiedfollmer} to more general
semimartingale financial models. We will assume that the condition of
(NFLVR-S) (see Theorem \ref{maintheorem}) holds. Recent works (see,
e.g., \cite{platen} and \cite{ruf}) have shown that in order to
find suitable trading strategies the condition of (NFLVR-S) can be
weakened and the hedging problem can be studied in markets that admit
certain types of arbitrage.

\subsubsection{Super-replication}\label{sec4.1.1}
Regarding the super-replication of contingent claims in markets with
short sales prohibition we have the following theorem.

\begin{teor}\label{theoremsuperreplication}
Suppose $\mathcal{M}_{\mathrm{sup}}(S)\neq\varnothing$. A nonnegative random
variable $f$ measurable with respect to $\mathcal{F}_T$ can be written as
%
\begin{equation}
\label{supereplieq}f=x+(H\cdot S)_T-C_T
\end{equation}
with $x$ constant, $H\in\mathcal{A}$ and $C\geq0$ an adapted and
nondecreasing c\`adl\`ag process with $C_0=0$ if and only if
\[
\sup_{Q\in\mathcal{M}_{\mathrm{sup}}(S)}E^Q[f]<\infty.
\]
In this case, $x=\sup_{Q\in\mathcal{M}_{\mathrm{sup}}(S)}E^Q[f]$ is the minimum
amount of initial capital for which there exist $H\in\mathcal{A}$ and
$C\geq0$ an adapted and nondecreasing c\`adl\`ag process with $C_0=0$
such that (\ref{supereplieq}) holds.
\end{teor}
\begin{pf}
This follows directly from Corollary \ref{thesetps} in this paper and
Examples~2.2, 4.1 and Proposition 4.1 in \cite{follmerkramkov}.
\end{pf}
Before we give an analogous result regarding perfect replication of
contingent claims, we present an example of a contingent claim that
cannot be super-replicated under short sales prohibition.
%
\begin{example}
This example illustrates how, under certain market hypotheses, it is
possible to explicitly exhibit a payoff that cannot be super-replicated
without short selling. Suppose that $S$ is of the form $S=\mathcal
{E}(R)$. Suppose that $R$ is a continuous $P$-martingale such that
$R_0=0$ and there exist $\varepsilon,C>0$ such that $P(\varepsilon\leq
[R,R]_T\leq C)=1$. Let $f=\exp(-R_T)$. We have, by Novikov's criterion
(see Theorem III-45 in \cite{protter}) and by Girsanov's theorem (see
Theorem III-40 in~\cite{protter}), that for every $\alpha>0$, $\frac
{dQ^{\alpha}}{dP}=\mathcal{E}(-\alpha R)_T$ defines a measure
$Q^{\alpha
}\in\mathcal{M}_{\mathrm{sup}}(S)$. Additionally,
\begin{eqnarray*}
E^{Q^{\alpha}}[f]&=&E^P\bigl[\mathcal{E}(-\alpha R)_Tf
\bigr]
\\
&=&E^P\bigl[\mathcal{E}\bigl(-(1+\alpha) R\bigr)_T\exp
\bigl((1/2+\alpha)[R,R]_T\bigr)\bigr]
\\
&\geq& E^P\bigl[\mathcal{E}\bigl(-(1+\alpha) R\bigr)_T
\bigr]\exp\bigl((1/2+\alpha)\varepsilon\bigr)
\\
&=&\exp\bigl((1/2+\alpha)\varepsilon\bigr)\rightarrow\infty
\end{eqnarray*}
as $\alpha$ goes to infinity. Hence $\sup_{Q\in\mathcal
{M}_{\mathrm{sup}}(S)}E^Q[f]=\infty$, and Theorem \ref{theoremsuperreplication}
implies that $f$ cannot be super-replicated without selling $S$ short.
However, if we assume that the market where $S$ can be sold short is
complete under $P$, that is, $\mathcal{M}_{\mathrm{loc}}(S)=\{P\}$, then in the
market where $S$ can be sold short $f$ can be replicated because it
belongs to $L^1(P)$. Indeed, we have that
\[
0\leq f\leq\exp\biggl(\frac{C}{2} \biggr)\mathcal{E}(-R)_T
\in L^1(P).
\]
\end{example}
%
\subsubsection{Replication}\label{sec4.1.2}
A question that remains open, however, is whether there exists a
characterization of contingent claims that can be perfectly replicated.
In this regard we have the following result analogous to the one proven
by Ansel and Stricker in \cite{anselstricker}; see also Theorems 5.8.1
and 5.8.4 in~\cite{KS}.

\begin{teor}\label{replication}
Suppose $\mathcal{M}_{\mathrm{sup}}(S)\neq\varnothing$. For a nonnegative random
variable $f$ measurable with respect to $\mathcal{F}_T$ the following
statements are equivalent:
\begin{longlist}[(ii)]
\item[(i)] $f=x+(H\cdot S)_T$ with $x$ constant and $H\in\mathcal{A}$
such that $(H\cdot S)$ is an $R^*$-martingale for some $R^*\in\mathcal
{M}_{\mathrm{sup}}(S)$.
\item[(ii)] There exists $R^*\in\mathcal{M}_{\mathrm{sup}}(S)$ such that
%
\begin{equation}
\label{maxattain}\sup_{Q\in\mathcal
{M}_{\mathrm{sup}}(S)}E^Q[f]=E^{R^*}[f]<
\infty.
\end{equation}
\end{longlist}
\end{teor}
\begin{pf}
That (i) implies (ii) follows from the fact that $(H\cdot S)$ is a
$Q$-su\-permartingale starting at 0 for all $Q\in\mathcal{M}_{\mathrm{sup}}(S)$;
see Corollary \ref{thesetps}. To prove that (ii) implies (i) we define
for all $t$ in $[0,T]$
%
\begin{equation}
\label{superrep}V_t:=\mathop{\operatorname{ess}\sup}_{Q\in\mathcal
{M}_{\mathrm{sup}}(S)}E^Q[f|
\mathcal{F}_t].
\end{equation}
By Lemma A.1 in \cite{follmerkramkov} the process $V$ is a
supermartingale under any $Q\in\mathcal{M}_{\mathrm{sup}}(S)$. In particular $V$
is an $R^*$-supermartingale. The fact that $V_T=f$ and (\ref
{maxattain}) imply that $V_0=E^{R^*}[V_T]$ and $V$ is a martingale under
$R^*$. On the other hand by Theorem 3.1 in \cite{follmerkramkov},
$V=V_0+(H\cdot S)-C$ for some $H\in\mathcal{A}$ and $C\geq0$
nondecreasing. Since $(H\cdot S)$ is an $R^*$-supermartingale (see
Corollary~\ref{thesetps}) we conclude that
\[
E^{R^*}[C_T]=V_0+E^{R^*}\bigl[(H\cdot
S)_T\bigr]-E^{R^*}[V_T]\leq0.
\]
Then, $C\equiv0$ $R^*$-almost surely and $(H\cdot S)$ is an $R^*$-martingale.
\end{pf}
$V_t$ in (\ref{superrep}) is usually used to define the selling price
of the claim $f$ at time~$t$. It represents the minimum cost of
super-replication of the claim $f$ at time $t$; see Proposition 4.1 in
\cite{follmerkramkov}. We now give an example of a payoff in markets
with continuous price processes which cannot be attained with
``martingale strategies.''
%
\begin{example}\label{exampledigitaloption}
Suppose that the market consists of a single risky asset with
continuous price process $S$. Assume further that $S$ is a
$P$-martingale which is not constant $P$-almost surely. Then $f=1_{\{
S_T\leq S_0\}}$ does not belong to the space
%
\begin{eqnarray}
\label{setg1}
&&\mathcal{G}:=\bigl\{x+(H\cdot S)_T\dvtx x\in\mathbb{R},H
\in\mathcal{A},
\nonumber\\[-8pt]\\[-8pt]
&&\hspace*{27.6pt}(H\cdot S)\mbox{ is a $Q$-martingale for some $Q\in\mathcal{M}_{\mathrm{sup}}(S)$}
\bigr\}.\nonumber
\end{eqnarray}
Indeed, for each $n\in\mathbb{N}$, let $(T_{n,m})_m$ be a localizing
sequence for
\[
\mathcal{E}\bigl(-n(S_t-S_0)\bigr).
\]
Define $Q_{n,m}\in\mathcal{M}_{\mathrm{sup}}(S)$ by
\[
\frac{dQ_{n,m}}{dP}=\mathcal{E}\bigl(-n(S_{T\wedge
T_{n,m}}-S_0)\bigr).
\]
We have that
\begin{eqnarray*}
E^{Q_{n,m}}[f]&=&1-E^{Q_{n,m}}[1-f]
\\
&=&1-E^P \biggl[1_{\{S_T> S_0\}}\exp\biggl(-n(S_{T\wedge T_{n,m}}-S_0)-
\frac
{n^2}{2}[S,S]_{T\wedge T_{n,m}} \biggr) \biggr].
\end{eqnarray*}
Since the expression under the last expectation is dominated by $\exp
(nS_0)\in\mathbb{R}$, the Dominated Convergence theorem implies that
for fixed $n$
\[
\lim_mE^{Q_{n,m}}[f]=1-E^P
\biggl[1_{\{S_T> S_0\}}\exp\biggl(-n(S_{T}-S_0)-
\frac{n^2}{2}[S,S]_{T} \biggr) \biggr].
\]
Applying the dominated convergence theorem once again we obtain that
\[
\lim_n\lim_mE^{Q_{n,m}}[f]=1.
\]
This allows us to conclude that
\[
\sup_{Q\in\mathcal{M}_{\mathrm{sup}}(S)}E^Q[f]=1.
\]
However, since $f$ is not $P$-almost surely constant, this supremum is
never attained. By Theorem \ref{replication}, $f$ does not belong to
the set $\mathcal{G}$ defined in (\ref{setg1}).
\end{example}
%
\begin{rem}
Example \ref{exampledigitaloption} illustrates that in nontrivial
markets with continuous price processes, the minimum super-replicating
cost of a digital option of the form $1_{\{S_T\leq S_0\}}$ is 1; See
Theorem \ref{theoremsuperreplication}. We will give other examples of
claims that cannot be perfectly replicated with martingale strategies
at the end of this section.
\end{rem}
We now proceed to give an alternative characterization of the random
variables in $\mathcal{G}$, with $\mathcal{G}$ as in (\ref{setg1}), by
extending the concept of maximal claims introduced by Delbaen and
Schachermayer in \cite{delbaenmain} and \cite{delbaennumeraire}.

\subsection{Maximal claims}\label{sec4.2}
By using the extension of the (FTAP) proved in Section \ref{sec3}, this section
generalizes the ideas presented in \cite{delbaennumeraire} to markets
with short sales prohibition. For simplicity, we assume below that $S$,
the price process of the underlying asset, is one-dimensional. The
results can be easily extended to the multi-dimensional case. Recall
the definitions of no arbitrage under short sales prohibition (NA-S)
and no free lunch with vanishing risk under short sales prohibition
(NFLVR-S) given in Section \ref{sec2}.

\subsubsection{The main theorem}\label{sec4.2.1}
%
\begin{defi}
Let $\mathcal{J}\subset L^0(P)$. We say that an element $f$ is maximal
in $\mathcal{J}$ if:
\begin{longlist}[(ii)]
\item[(i)] $f\in\mathcal{J}$ and
\item[(ii)] $f\leq g$ $P$-almost surely and $g\in\mathcal{J}$ imply
that $f=g$ $P$-almost surely.
\end{longlist}
\end{defi}

\begin{defi}\label{thesetB}
Given $H\in\mathcal{A}$, we define $\mathcal{B}(H)$ as the set of
random variables of the form
\[
\bigl(\bigl(H^1,H^2\bigr)\cdot\bigl(S^1,S^2
\bigr)\bigr)_T,
\]
where $S^1=(H\cdot S)$, $S^2=S$; $(H^1,H^2)\in L(S^1,S^2)$; $H^2\geq
0$, $H^1_0\equiv1$, $H^2_0\equiv0$ and
%
\begin{equation}
\label{equationworkable}\bigl(H^1-1,H^2\bigr)\cdot
\bigl(S^1,S^2\bigr)\geq-\beta-\alpha S^1
\end{equation}
for some $\alpha,\beta>0$.
\end{defi}
The following is the main theorem of this section.
%
\begin{teor}\label{maximalclaims}
Let $f\in L^0(P)$ be a random variable bounded from below. The
following statements are equivalent:
\begin{enumerate}[(iii)]
\item[(i)] $f=(H\cdot S)_T$ for some $H\in\mathcal{A}$ such that:
\begin{enumerate}[(a)]
\item[(a)] the market where $S^1=(H\cdot S)$ and $S^2=S$ trade with
short selling prohibition on $S^2$ satisfies (NFLVR-S) and
\item[(b)] $f$ is maximal in $\mathcal{B}(H)$ (see Definition \ref{thesetB}).
\end{enumerate}
\item[(ii)] There exists $R^*\in\mathcal{M}_{\mathrm{sup}}(S)$ such that
\[
\sup_{Q\in\mathcal{M}_{\mathrm{sup}}(S)}E^Q[f]=E^{R^*}[f]=0.
\]
\item[(iii)] There exists $H\in\mathcal{A}$ such that $f=(H\cdot S)_T$
and $(H\cdot S)$ is an $R^*$-martingale for some $R^*$ in $\mathcal
{M}_{\mathrm{sup}}(S)$.
\end{enumerate}
If we further assume that $f$ is bounded and $\mathcal{M}_{\mathrm{loc}}(S)\neq
\varnothing$, the above statements are equivalent to:
\begin{enumerate}[(iii)]
\item[(iv)] There exists $H\in\mathcal{A}$ such that $f=(H\cdot S)_T$
for some $H\in\mathcal{A}$ and $(H\cdot S)$ is an $R$-martingale for
all $R$ in $\mathcal{M}_{\mathrm{loc}}(S)$.
\end{enumerate}
\end{teor}
%
\begin{rem}
It is important to point out that we can take the same measure $R^*$ in
(ii) and (iii), and the same strategy $H$ in (i), (iii) and (iv).
\end{rem}
Before establishing some lemmas necessary to prove this theorem we make
some additional remarks.
%
\begin{rem}
A related result for diffusion price processes can be found in Theorem
5.8.4 in \cite{KS}. This theorem uses the alternative assumption that
\[
\bigl\{(H\cdot S)_{\rho}\mbox{: $\rho$ is a stopping time in $[0,T]$}
\bigr\}
\]
is $Q$-uniformly integrable for all $Q\in\mathcal{M}_{\mathrm{sup}}(S)$. This
hypothesis also implies that $(H\cdot S)$ is a $Q$-martingale for all
$Q\in\mathcal{M}_{\mathrm{loc}}(S)$.
\end{rem}
%
\begin{rem}
Condition (\ref{equationworkable}) resembles the definition of workable
contingent claims studied in \cite{delbaenworkable}.
\end{rem}
%
\begin{rem}\label{remarkmaximaliywithoutconstraints}
If $f=(H\cdot S)_T$, $(H\cdot S)$ is an $R^*$-martingale for some
$R^*\in\mathcal{M}_{\mathrm{sup}}(S)$ and $1_{\{H=0\}}\cdot S$ is
indistinguishable from 0, then $R^*\in\mathcal{M}_{\mathrm{loc}}(S)\neq
\varnothing
$. Indeed, observe that if we call $M=(H\cdot S)$, then, by Corollary
3.5 in \cite{anselstricker}, $ (\frac{1}{H}1_{\{H\neq0\}}
)\cdot M=1_{\{H\neq0\}}\cdot S=S-S_0$ is an $R^*$-local martingale.
Theorem 11.4.4 in \cite{delbaenbook} implies that the claim $f$ is
also maximal in
%
\begin{equation}
\label{admissiblewithnoshortsalesproh}\tilde{\mathcal{K}}=\bigl\{
(H\cdot
S)_T\dvtx H\in\tilde{A}\bigr\},
\end{equation}
where $\tilde{A}$ is the set of strategies that satisfy (i), (ii) and
(iii) in Definition \ref{admissiblestrategies}. Additionally, also by
Theorem 11.4.4 in \cite{delbaenbook}, Theorem \ref{maximalclaims}
shows that when $\mathcal{M}_{\mathrm{loc}}(S)\neq\varnothing$, all bounded
maximal claims in $\mathcal{B}(H)$ (see Definition \ref{thesetB}) of
the form $(H\cdot S)_T$ for some $H\in\mathcal{A}$ are maximal in
$\tilde{\mathcal{K}}$ as defined in (\ref{admissiblewithnoshortsalesproh}).
\end{rem}
The proof of Theorem \ref{maximalclaims} that we present below mimics
the argument presented in \cite{delbaennumeraire}. In this
generalization, the (FTAP) under short sales prohibition (Theorem~\ref
{maintheorem}) and the results presented by Kabanov in \cite{kabanov}
are fundamental.

\subsubsection{Some lemmas}\label{sec4.2.2}
We first recall the following definition.
%
\begin{defi}
A subset $\mathcal{N}$ of $L^0(P)$ is bounded in $L^0(P)$ if for all
$\varepsilon>0$ there exists $M>0$ such that $P(|Y|>M)<\varepsilon$ for all
$Y\in\mathcal{N}$.
\end{defi}
The following lemmas will be used.
%
\begin{lema}
The condition of (NFLVR-S) holds if and only if (NA-S) holds, and the
set
\[
\mathcal{K}_1=\bigl\{(H\cdot S)_T\mbox{: $H\in
\mathcal{K}$ and $(H\cdot S)\geq-1$}\bigr\}
\]
is bounded in $L^0(P)$.
\end{lema}
\begin{pf}
This corresponds to Lemma 2.2 in \cite{kabanov}. As already noticed
before in the proof of Theorem \ref{maintheorem}, the results in \cite
{kabanov} can be applied to our case because the convex portfolio
constraints satisfy the desired hypotheses.
\end{pf}
%
\begin{lema}\label{lemalocalmartingalefactor}
The condition of (NFLVR-S) holds if and only if (NA-S) holds and there
exists a strictly positive $P$-local martingale $L=(L_t)_{0\leq t\leq
T}$ such that $L_0=1$ and $P\in\mathcal{M}_{\mathrm{sup}}(LS)$.
\end{lema}
\begin{pf}
The same proof of Theorem 11.2.9 in \cite{delbaenbook} can be applied
to our context.
\end{pf}
We now state from our framework a result that is analogous to Theorem
11.4.2 in \cite{delbaenbook}. This theorem gives necessary and
sufficient conditions under which the condition of (NA-S) holds after a
change of num\'eraire. We will need the following lemma, that proves
that the self-financing condition [see~(\ref{selffinancing})] is
independent of the choice of num\'eraire; see also \cite{protterintroduction}.
%
\begin{lema}\label{lemmaselffinancing}
Let $V$ be a positive $P$-semimartingale, $M= (\frac{S}{V},\frac
{1}{V},1 )$ and $N=(S,1,V)$. For a (three-dimensional) predictable
process $H$ the following statements are equivalent:
\begin{longlist}[(ii)]
\item[(i)] $H\in L(M)$ and
\[
H\cdot M=HM-H_0M_0=H^1\frac{S}{V}+H^2
\frac{1}{V}+H^3-H^1_0
\frac
{S_0}{V_0}-H^2_0\frac{1}{V_0}-H^3_0;
\]
\item[(ii)] $H\in L(N)$ and
\[
H\cdot N=HN-H_0N_0=H^1S+H^2+H^3V-H^1_0S_0-H^2_0-H^3_0V_0.
\]
\end{longlist}
\end{lema}
\begin{pf}
($\Rightarrow$) Let $W=H\cdot M$. By (i), $\Delta W=H\Delta M=HM-HM_-$
and $W_-=W-\Delta W=HM_--H_0M_0$. The integration by parts formula
implies that
\begin{eqnarray*}
d(VW)&=&W_-\,dV+V_-\,dW+d[W,V]
\\
&=&(HM_--H_0M_0)\,dV+V_-H\,dM+d[W,V].
\end{eqnarray*}
Since $d[W,V]=H\,d[M,V]$ regrouping terms and using integration by parts
once more we obtain that
\begin{eqnarray*}
d(VW)&=&H\bigl(M_-\,dV+V_-\,dM+d[M,V]\bigr)-H_0M_0\,dV
\\
&=&H\,d(VM)-H_0M_0\,dV.
\end{eqnarray*}
We have that $VM=N$, and hence $d(VW)=H\,dN-H_0M_0\,dV$. By (i),
$VW=HN-VH_0M_0$ and
\begin{eqnarray*}
H\,dN&=&d(VW)+H_0M_0\,dV
\\
&=&\bigl(d(HN)-H_0M_0\,dV\bigr)+H_0M_0\,dV
\\
&=&d(HN)
\end{eqnarray*}
as we wanted to show.\vadjust{\goodbreak}

($\Leftarrow$) The proof of this direction is analogous to the one just
presented since $M$ is obtained after multiplying $N$ by the
nonnegative semimartingale~$\frac{1}{V}$.
\end{pf}
%
\begin{lema}\label{lemmanoarbitragenumeraire}
Suppose that $V$ is a strictly positive $P$-semimartingale. The market
with multi-dimensional price process $ (\frac{1}{V},\frac{S}{V}
)$, where short selling prohibition is imposed on $\frac{S}{V}$,
satisfies the condition of (NA-S) if and only if $V_T-V_0$ is maximal
in $\mathcal{D}$, where $\mathcal{D}$ is the set of random variables of
the form $(H\cdot(S,V))_T$ where $H^1\geq0$, $H^1_0\equiv0$,
$H^2_0\equiv1$ and
\[
\bigl(H^1,H^2-1\bigr)\cdot(S,V)\geq-\alpha V\qquad\mbox{for
some $\alpha> 0$}.
\]
\end{lema}
\begin{pf} ($\Leftarrow$) Let $M= (\frac{1}{V},\frac{S}{V}
)$ and $N=(S,V)$. Suppose that $H=(H^1,H^2)$ is an arbitrage in the
market with multi-dimensional price process $ (\frac{1}{V},\frac
{S}{V} )$. In other words, assume that $H^2\geq0$, $H_0\equiv0$,
$(H\cdot M)_T\geq0$, $P((H\cdot M)_T>0)>0$ and $H\cdot M\geq-\alpha$
for some $\alpha>0$. If we define
\begin{eqnarray*}
H^3&=&1+H\cdot M-HM,
\\
\tilde{M} &=& \biggl(\frac{1}{V},\frac{S}{V},1 \biggr),
\\
\tilde{N}&=&(1,S,V)
\end{eqnarray*}
and
\[
\tilde{H}=\bigl(H^1,H^2,H^3\bigr),
\]
we have that $\tilde{H}\cdot\tilde{M}=\tilde{H}\tilde{M}-1$. By Lemma
\ref{lemmaselffinancing} we have that
\[
\tilde{H}\cdot\tilde{N}=\tilde{H}\tilde{N}-V_0.
\]
But observe that
\[
\tilde{H}\tilde{N}=VHM+(1+H\cdot M-HM)V=(1+H\cdot M)V
\]
and
\[
\tilde{H}\cdot\tilde{N}=K\cdot N,
\]
where $K=(H^2,H^3)$. Hence $(K\cdot N)_T$ is an element of $\mathcal
{D}$ such that $(K\cdot N)_T\geq V_T-V_0$ $P$-almost surely and
$P((K\cdot N)_T>V_T-V_0)>0$, whence $V_T-V_0$ is not maximal in
$\mathcal{D}$.

($\Rightarrow$) Conversely, suppose that $V_T-V_0$ is not maximal in
$\mathcal{D}$. With the notation used above, let $K=(K^1,K^2)$ be a
strategy such that $(K\cdot N)_T\geq V_T-V_0$ $P$-almost surely and
$P((K\cdot N)_T>V_T-V_0)>0$, with $K^1\geq0$, $K^1_0\equiv0$,
$K^2_0\equiv1$ and $(K^1,K^2-1)\cdot N\geq-\alpha V$ for some $\alpha
>0$. Define $H^2=K^1$, $H^3=K^2-1$, $H^1=(H^2,H^3)\cdot N-(H^2,H^3)N$
and $H=(H^1,H^2,H^3)$. We have that $H\cdot\tilde{N}=H\tilde
{N}-H_0\tilde{N}_0$. By Lemma \ref{lemmaselffinancing} we have that
\[
H\cdot\tilde{M}=H\tilde{M}-H_0\tilde{M}_0=H\tilde{M}.
\]
Hence,
\[
\bigl(H^1,H^2\bigr)\cdot M=H\tilde{M}.
\]
We have that
\[
H\tilde{M}=\frac{1}{V}H\tilde{N}=\frac{1}{V} \bigl(
\bigl(H^2,H^3\bigr)\cdot N \bigr)=\frac{1}{V}
\bigl(K\cdot N-(V-V_0)\bigr)\geq-\alpha.
\]
Therefore,
\[
\bigl(\bigl(H^1,H^2\bigr)\cdot M\bigr)_T=
\frac{1}{V_T}\bigl((K\cdot N)_T-(V_T-V_0)
\bigr),
\]
$((H^1,H^2)\cdot M)_T\in L^0_+(P)$ and $P(((H^1,H^2)\cdot M)_T>0)>0$.
Since \mbox{$H^1_0=H^2_0=0$}, $(H^1,H^2)$ is an arbitrage strategy in the
market with multi-dimensional price process $ (\frac{1}{V},\frac
{S}{V} )$.
\end{pf}
%
\begin{rem}
It is important to observe that the no arbitrage condition \mbox{(NA-S)} over
$ (\frac{1}{V},\frac{S}{V} )$ holds for strategies that are
nonnegative on the second component, but can be negative in an
admissible way [see condition (iii) in Definition~\ref
{admissiblestrategies}] over the first component. Lemma \ref
{lemmanoarbitragenumeraire} gives a necessary and sufficient condition
under which the introduction of $V$ as a num\'eraire does not introduce
arbitrage in a market with short sales prohibition. A related
discussion on num\'eraires over convex sets of random variables can be
found in \cite{kostaspositivity}.
\end{rem}
These lemmas allow us to prove Theorem \ref{maximalclaims}.

\subsubsection{Proof of the main theorem}\label{sec4.2.3}

\mbox{}

\begin{pf*}{Proof of Theorem \ref{maximalclaims}}
Since $f$ is bounded from below there exists a constant $x$ such that
$\tilde{f}:=f+x$ is nonnegative. Theorem \ref{replication}, applied to
$\tilde{f}$, proves the equivalence between (ii) and (iii). We will
prove now that (iii) implies (i). The (FTAP) (Theorem \ref
{maintheorem}) shows that (NFLVR-S) holds for the market consisting of
$S$ and $(H\cdot S)$ with short selling prohibition on $S$. Now assume
that $f\leq((H^1,H^2)\cdot(S^1,S^2))_T$ with $((H^1,H^2)\cdot
(S^1,S^2))_T\in\mathcal{B}(H)$. Then
\[
\bigl(H^1-1,H^2\bigr)\cdot\bigl(S^1,S^2
\bigr)\geq-\beta-\alpha S^1
\]
for some $\alpha,\beta>0$ and $((H^1-1,H^2)\cdot(S^1,S^2))_T\geq0$.
Since
\[
\bigl(H^1-1+\alpha,H^2\bigr)\cdot\bigl(S^1,S^2
\bigr)\geq-\beta
\]
by Lemma \ref{lemma2} (extended to the case when the integrand is not
identically 0 at time~0) we conclude that
\[
\bigl(H^1-1+\alpha,H^2\bigr)\cdot\bigl(S^1,S^2
\bigr)
\]
is an $R^*$-supermartingale, which in turn implies that
$((H^1-1,H^2)\cdot(S^1,S^2))$ is an $R^*$-supermartingale starting at
0. Since $((H^1-1,H^2)\cdot(S^1,S^2))_T\geq0$, we conclude that
$((H^1-1,H^2)\cdot(S^1,S^2))_T= 0$ $P$-almost surely. This shows that
$f$ is maximal in $\mathcal{B}(H)$.

Let us prove now that (i) implies (iii). By the (FTAP) we know that
there exists $\tilde{P}\in\mathcal{M}_{\mathrm{sup}}(S)$ such that $(H\cdot S)$
is a $\tilde{P}$-local martingale. Let $a$ be such that $V:=a+(H\cdot
S)$ is positive and bounded away from 0. Since $f$ is maximal in
$\mathcal{B}(H)$, $V_T-V_0$ is maximal in $\mathcal{D}$, where
$\mathcal
{D}$ is as in Lemma \ref{lemmanoarbitragenumeraire}.
By Lemma \ref{lemmanoarbitragenumeraire} (NA-S) holds in the market
where $\frac{S}{V}$ and $\frac{1}{V}$ trade with short selling
prohibition on $\frac{S}{V}$. By Lemma \ref{lemalocalmartingalefactor}
we conclude that (NFLVR-S) holds in this market with respect to the
measure $\tilde{P}$. Hence, by the (FTAP) there exists $\tilde{Q}\sim
\tilde{P}$ (and hence $\tilde{Q}\sim P$) such that $\frac{S}{V}$ is a
$\tilde{Q}$-supermartingale, and $\frac{1}{V}$ is a bounded $\tilde
{Q}$-local martingale and therefore a $\tilde{Q}$-martingale. By
defining $R^*$ by $V_T\,dR^*= (E^{\tilde{Q}} [\frac{1}{V_T}
] )^{-1}\,d\tilde{Q}$, we observe that $R^*\in\mathcal{M}_{\mathrm{sup}}(S)$
and $V$ is an $R^*$-martingale. This implies that $(H\cdot S)$ is an
$R^*$-martingale as well.

Finally to prove that (iii) implies (iv) we observe that if $R\in
\mathcal{M}_{\mathrm{loc}}(S)$ and $(\tau_n)$ is an $R$-localizing sequence for
$(H\cdot S)$ then $(H\cdot S)_{\tau_n\wedge T}=E^{R^*}[f|\mathcal
{F}_{\tau_n\wedge T}]$ is a dominated sequence of random variables with
zero $R$-expectation. By the dominated convergence theorem we conclude
that $E^R[f]=0$, and $(H\cdot S)$ is an $R$-martingale (it is an
$R$-supermartingale with constant expectation).
\end{pf*}

\subsection{Final remarks}\label{sec4.3}
%
\begin{rem}
Condition (i) in Theorem \ref{maximalclaims} can be interpreted as
follows. The market where $S^1$ and $S^2$ trade with short sales
prohibition on $S^2$ satisfies the no arbitrage paradigm of (NFLVR-S).
In this market the strategy of buying and holding $S^1$ cannot be
dominated by any strategy with initial holdings of one share of $S^1$
and none of $S^2$ that does not sell $S^2$ short.
\end{rem}
The following observation is important. It shows that the elements
$f\in L^0(P)$ that satisfy any of the conditions of Theorem \ref
{maximalclaims} are maximal in $\mathcal{K}$. A~related result was
discussed in Remark \ref{remarkmaximaliywithoutconstraints}, where it
was shown that, under stronger assumptions on the replicating strategy
for $f$, a stronger form of maximality holds, namely maximality in
$\tilde{\mathcal{K}}$; see (\ref{admissiblewithnoshortsalesproh}).
%
\begin{prop}
If \textup{(i)}, \textup{(ii)} or \textup{(iii)} in Theorem \ref{maximalclaims} holds, then $f$ is
maximal in $\mathcal{K}$.
\end{prop}
\begin{pf}
Assume that $E^{R^*}[f]=0$ for some $R^*\in\mathcal{M}_{\mathrm{sup}}(S)$. If
$f\leq(K\cdot S)_T$ with $K\in\mathcal{A}$, by Lemma \ref{lemma2}, we
conclude that $E^{R^*}[(K\cdot S)_T]=0$. This implies that $f=(K\cdot
S)_T$ $P$-almost surely and $f$ is maximal in $\mathcal{K}$.
\end{pf}
As shown in \cite{delbaencounterexample}, without the assumption
that $f$ is bounded, (iv) of Theorem~\ref{maximalclaims} is not a
necessary condition. Theorem \ref{maximalclaims} is useful to argue why
certain types of contingent claims in certain financial models cannot
be replicated by using a strategy that is maximal in the sense of (i)
of Theorem \ref{maximalclaims} above.
%
\begin{example}\label{exampleputoption}
Let $K\in(0,\infty)$ be fixed. Assume that $S$ is a continuous
$P$-martingale, $[S,S]_T$ is deterministic and $P(S_T<K,\tau<T)>0$ where
\[
\tau=\inf\bigl\{t\leq T\dvtx S_t\geq K+\tfrac{1}{2}
\bigl([S,S]_T-[S,S]_t \bigr) \bigr\}\wedge T.
\]
By Novikov's criterion (Theorem III-45 in \cite{protter}) and by
Girsanov's theorem (Theorem III-40 in \cite{protter}) we know that
\[
\frac{dQ}{dP}=\mathcal{E} \biggl(-\int_0^T1_{[\tau,T]}(s)
\,dS_s \biggr)
\]
defines a probability measure $Q\in\mathcal{M}_{\mathrm{sup}}(S)$. If
$g\dvtx [0,\infty)\rightarrow[0,\infty)$ is a function that vanishes on
$[K,\infty)$ and is strictly positive on $[0,K)$, then
%
\begin{eqnarray}\label{inequalityexampleputoption}
E^Q\bigl[g(S_T)\bigr]&=&E^Q
\bigl[g(S_T)1_{\{S_T<K\}}\bigr]
\nonumber
\\
&\geq& E^P\bigl[1_{\{\tau=T\}}g(S_T)1_{\{S_T<K\}}
\bigr]
\nonumber\\[-8pt]\\[-8pt]
&&{} +E^P \bigl[1_{\{S_T<K,\tau<T\}}g(S_T)\exp
\bigl(-(S_T-K)\bigr) \bigr]
\nonumber
\\
&>&E^P\bigl[g(S_T)\bigr].\nonumber
\end{eqnarray}
If we further assume that $g$ is bounded, then by Theorem \ref
{maximalclaims} [condition (iv)] we conclude that $g(S_T)$ does not
belong to $\mathcal{G}$ as in (\ref{setg1}). Indeed, if
$g(S_T)=x+(H\cdot S)_T$, with $x\in\mathbb{R}$, $H\in\mathcal{A}$ and
$(H\cdot S)$ an $R^*$-martingale for some $R^*\in\mathcal{M}_{\mathrm{sup}}(S)$,
then by Theorem \ref{maximalclaims}, $(H\cdot S)$ would be an
$R$-martingale for all $R\in\mathcal{M}_{\mathrm{loc}}(S)$. In particular, we
would have that $E^P[g(S_T)]=x=E^Q[g(S_T)]$, which contradicts~(\ref
{inequalityexampleputoption}).
The function $g(x)=(K-x)_+$ satisfies the above mentioned conditions.
Hence under these assumptions, the put option's payoff does not belong
to $\mathcal{G}$. This example is similar to Example 7.2 of \cite
{cvitanickaratzas}.
\end{example}
%
\begin{rem}
In Example 5.7.4 in \cite{KS} and Section 8.1 in \cite{cvitanic}, it is
proven that for diffusion models with constant coefficients and
stochastic volatility models with additional properties, respectively,
the minimum super-replication price of an European put option, $\sup
_{Q\in\mathcal{M}_{\mathrm{sup}}(S)}E^Q[(K-S_T)_+]$,
is equal to $K$. In particular if $P(S_T\neq0)>0$, then this supremum
is never attained and $(K-S_T)_+$ is not in $\mathcal{G}$ as defined by
(\ref{setg1}).
\end{rem}
In this section we have studied the space of contingent claims that can
be super-replicated and perfectly replicated with martingale strategies
in a market with short sales prohibition. We extended results found in
\cite{anselstricker,KS} and \cite{schiedfollmer} to the short
sales prohibition case. We additionally have extended the results in
\cite{delbaennumeraire} to our framework and modified the concept of
maximality accordingly (see Theorem \ref{maximalclaims}). Additionally,
we presented explicit payoffs in general markets that cannot be
replicated without selling the spot price process short.

\section{Open questions}\label{sec5}
It is still unclear whether (NFLVR) for a market without short sales
prohibition, implies that all claims that are maximal in the sense of
(i) in Theorem \ref{maximalclaims} are maximal in
$\tilde{\mathcal{K}}$; see (\ref{admissiblewithnoshortsalesproh}).
Equivalently, it is unclear whether for a claim $f$ that is bounded
from below, the conditions $\mathcal{M}_{\mathrm{loc}}(S)\neq\varnothing$ and
\[
\sup_{Q\in\mathcal{M}_{\mathrm{sup}}(S)}E^Q[f]=E^{R^*}[f]
\]
for some
$R^*\in\mathcal{M}_{\mathrm{sup}}(S)$, imply that there exists
$P^*\in\mathcal{M}_{\mathrm{loc}}(S)$ such that $E^{P^*}[f]=E^{R^*}[f]$. Also,
it would be interesting to obtain a characterization of the set of
claims that are maximal in $\mathcal{K}$ [as in (\ref{setk})] and
explore whether maximality in $\mathcal{K}$ implies maximality in
$\tilde{\mathcal{K}}$; see (\ref{admissiblewithnoshortsalesproh}).

\section*{Acknowledgments}

This paper comprises a large part of the author's Ph.D.
thesis written under the direction of Philip Protter and Robert Jarrow
at Cornell University. The author wishes to thank them for the help
they provided during his graduate career. Also, the author would like
to thank Kasper Larsen, Martin Larsson, Alexandre Roch and Johannes Ruf
for their comments and insights into this work. Special thanks go to an
anonymous referee for pertinent remarks and corrections on earlier
versions of the manuscript.



\printaddresses

\end{document}